\documentclass[12pt]{article}
\usepackage{enumerate}
\usepackage{amsfonts}
\usepackage{amsmath}
\usepackage{amssymb}
\usepackage{amsthm}
\usepackage{latexsym}
\usepackage{color}
\usepackage{tcolorbox}

\def\H{\mathcal{H}}
\def\HH{\mathcal{H}}

\def\D{\mathfrak{D}}
\def\C{\mathfrak{C}}
\def\T{\mathfrak{T}}

\def\B{\mathfrak{B}}
\def\N{\mathbb{N}}
\newcommand{\FF}{\mathfrak{F}}

\newcommand{\id}{\mathrm{Id}}
\newcommand{\Tr}{\mathrm{Tr}}

\newcommand{\shs}{\hspace{1pt}}
\newcounter{defin}  \newcounter{lemma}  \newcounter{theorem}
\newcounter{property} \newcounter{corol}  \newcounter{remark} \newcounter{example}
\newenvironment{lemma}{\par\refstepcounter{lemma}     \textbf{Lemma \thelemma.} }{\rm\par}
\newenvironment{theorem}{\par\refstepcounter{theorem}     \textbf{Theorem \thetheorem.}\ }{\rm\par}
\newenvironment{property}{\par\refstepcounter{property}     \textbf{Proposition \theproperty.}\ }{\rm\par}
\newenvironment{corollary}{\par\refstepcounter{corol}     \textbf{Corollary \thecorol.} }{\rm\par}

\newenvironment{remark}{\par\refstepcounter{remark}     \textbf{Remark \theremark.}}{\rm\par}

\begin{document}

\title{Continuity of the relative entropy of resource}

\author{Ludovico Lami,\thanks{Korteweg--de Vries Institute for Mathematics, Institute for Theoretical Physics, and QuSoft, University of Amsterdam, the Netherlands; e-mail:ludovico.lami@gmail.com}\quad Maksim Shirokov\thanks{Steklov Mathematical Institute, Moscow, Russia, e-mail:msh@mi-ras.ru}}

\maketitle

\begin{abstract}
A  criterion  of local continuity of the relative entropy of resource --- the relative entropy distance to the set of free states --- is obtained. Several basic corollaries of this criterion are presented. Applications to the relative entropy of entanglement in multipartite quantum systems are considered. It is shown, in particular, that local continuity of any relative entropy of multipartite entanglement follows from local continuity of the quantum mutual information.
\end{abstract}

\section{Introduction and preliminaries}

In recent years, the mathematical framework of quantum resource theories has emerged as a unifying paradigm to investigate the transformation of quantum resources~\cite{RT-review}. In this framework, quantum states can be either `free', meaning that they can be prepared relatively inexpensively with the equipment and operational constraints at hand, or `resourceful', meaning that they are difficult to obtain and therefore potentially exhibit some useful property that free states do not. One of the central problem in quantum resource theories is resource quantification: how to quantify the resource content of states that are not free? Among the several resource measures available, the relative entropy of resource~\cite{Vedral1997, tightuniform, achievability}, defined as the distance from the set of free states measured by the quantum relative entropy~\cite{Umegaki1962, Hiai1991}, has been identified as a key universal resource quantifier. In the paradigmatic case of entanglement theory, from whose blueprint the framework of quantum resource theories developed, this is nothing but the relative entropy of entanglement that has been the subject of much investigation~\cite{Vedral1997, Vedral1998, Donald1999, Plenio2000, MatthiasPhD, BrandaoPlenio2, Brandao2010, Piani2009, tightuniform, achievability}. More generally, the importance of the relative entropy of resource stems from the fact that its regularised version is believed to represent both the ultimate rate of conversion between different states under the constraints of the given resource theory, and the Stein exponent associated with hypothesis testing of resourceful states. This belief is encapsulated in the famous Generalised Quantum Stein's Lemma conjecture~\cite{BrandaoPlenio2, Datta-alias, Brandao-Gour, gap}.

In infinite-dimensional quantum information, many entropic quantifiers are well known to be ill-behaved and generally discontinuous on the set of all density operators~\cite{Wehrl}. Accordingly, infinite-dimensional quantum information poses significant technical challenges~\cite{HOLEVO-CHANNELS-2}. For resource theories, finding easy-to-check conditions that guarantee the continuity of resource measures in general is an important endeavor. In the case of the relative entropy of resource, we are further motivated to do so by our recent finding that this measure has surprisingly benign properties even in the infinite-dimensional case~\cite{achievability}, essentially derived from the key observation that the infimum over free states is often always achieved. When this holds, the relative entropy of resource is lower semicontinuous everywhere~\cite[Theorem~5]{achievability}. While encouraging, these findings tell only half of the story: on which sequences is the relative entropy of resource actually continuous, rather than merely lower semicontinuos?

In this work, we set out to answer this question, determining some sufficient conditions that guarantee the continuity of the relative entropy of resource on a given sequence of states. Our main results in this sense are Theorem~\ref{conv-th} and Corollary~\ref{conv-th-c-2} below. In Propositions~\ref{ER-conv-1} and~\ref{ER-conv-2} we then apply these general findings to the special case of the relative entropy of multipartite entanglement, one of the most important for applications.

\bigskip

We start by fixing some terminology. Let $\HH$ be a separable  Hilbert space,
$\B(\HH)$ the algebra of all bounded operators  on $\HH$ with the operator norm $\|\cdot\|$ and $\T(\HH)$ the
Banach space of all trace-class
operators on $\HH$ with the trace norm $\|\!\cdot\!\|_1$. Let
$\D(\mathcal{H})$ be  the set of quantum states (positive operators
in $\T(\HH)$ with unit trace)~\cite{HOLEVO-CHANNELS-2,MARK}.

Along with the norm topology on the Banach space $\T(\HH)$, we will consider the \emph{weak*-topology} induced on $\T(\HH)$ by its pre-dual space $\C(\HH)$ --
the Banach space of compact operators on $\HH$ with the operator norm. The weak*-topology can be defined as the coarsest topology making all functions of the form $\T(\HH)\ni \rho\mapsto \Tr C\rho$, $C\in \C(\HH)$, continuous  (see Section II.A. in~\cite{achievability} for more detail description of this topology).

The \emph{von Neumann entropy} of a quantum state
$\rho \in \D(\HH)$ is  defined by the formula
$S(\rho)=\Tr\eta(\rho)$, where  $\eta(x)=-x\ln x$ if $x>0$
and $\eta(0)=0$. It is a concave lower semicontinuous function on the set~$\D(\HH)$ taking values in~$[0,+\infty]$~\cite{PETZ-ENTROPY,HOLEVO-CHANNELS-2}.\smallskip

We will use the  homogeneous extension of the von Neumann entropy to the positive cone $\T_+(\HH)$ defined as
\begin{equation}\label{S-ext}
S(\rho)\doteq(\Tr\rho)S(\rho/\Tr\rho)=\Tr\eta(\rho)-\eta(\Tr\rho)
\end{equation}
for any nonzero operator $\rho$ in $\T_+(\HH)$ and equal to $0$ at the zero operator~\cite{Lindblad1974}.\smallskip

The \emph{quantum relative entropy} for two states $\rho$ and
$\sigma$ in $\D(\HH)$ is defined as
\begin{equation*}
D(\rho\shs\|\shs\sigma)=\sum_i\langle
\varphi_i|\,\rho\ln\rho-\rho\ln\sigma\,|\varphi_i\rangle,
\end{equation*}
where $\{\varphi_i\}$ is the orthonormal basis of
eigenvectors of the state $\rho$ and it is assumed that
$D(\rho\,\|\sigma)=+\infty$ if $\,\mathrm{supp}\rho\shs$ is not
contained in $\shs\mathrm{supp}\shs\sigma$~\cite{HOLEVO-CHANNELS-2,Lindblad1974}.\footnote{The support $\mathrm{supp}\rho$ of a state $\rho$ is the closed subspace spanned by the eigenvectors of $\rho$ corresponding to its positive eigenvalues.}\smallskip

We will use  Lindblad's extension of quantum relative entropy to any positive
operators $\rho$ and
$\sigma$ in $\mathfrak{T}(\mathcal{H})$ defined as
\begin{equation*}
D(\rho\shs\|\shs\sigma)=\sum_i\langle\varphi_i|\,\rho\ln\rho-\rho\ln\sigma+\sigma-\rho\,|\varphi_i\rangle,
\end{equation*}
where $\{\varphi_i\}$ is the orthonormal basis of
eigenvectors of the operator  $\rho$ and it is assumed that $\,D(0\|\shs\sigma)=\Tr\sigma\,$ and
$\,D(\rho\shs\|\sigma)=+\infty\,$ if $\,\mathrm{supp}\rho\shs$ is not
contained in $\shs\mathrm{supp}\shs\sigma$ (in particular, if $\rho\neq0$ and $\sigma=0$)
\cite{Lindblad1974}. If the extended von Neumann entropy $S(\rho)$ of $\rho$ (defined in~\eqref{S-ext}) is finite
then
\begin{equation}\label{re-exp}
D(\rho\shs\|\shs\sigma)=\Tr\rho(-\ln\sigma)-S(\rho)-\eta(\Tr\rho)+\Tr\sigma-\Tr\rho.
\end{equation}

The function $(\rho,\sigma)\mapsto D(\rho\shs\|\shs\sigma)$ is nonnegative lower semicontinuous and jointly convex on
$\T_+(\HH)\times\T_+(\HH)$. It is easy to show that
\begin{equation}\label{D-mul}
D(c\rho\shs\|\shs c\sigma)=cD(\rho\shs\|\shs \sigma),\qquad\qquad\qquad\qquad\quad\;
\end{equation}
and
\begin{equation}\label{D-c-id}
D(\rho\shs\|\shs c\sigma)=D(\rho\shs\|\shs\sigma)-\Tr\rho\ln c+(c-1)\Tr\sigma
\end{equation}
for any $\rho,\sigma\in\T_+(\HH)$ and $c\geq0$. \medskip

The \emph{quantum mutual information}  of a state $\rho$ of a $m$-partite quantum system $A_1...A_m$ is defined as
\begin{equation}\label{MI-m-def}
     I(A_1\!:\!...\!:\!A_m)_{\rho}\doteq
    D(\rho\shs\|\shs\rho_{A_{1}}\otimes...\otimes\rho_{A_{m}})=\sum_{k=1}^m S(\rho_{A_{k}})-S(\rho),
\end{equation}
where the second formula is well defined if $S(\rho)<+\infty$.  This quantity is treated as a measure of total
correlation of a state $\rho$, it is equal to zero if and only if $\rho$ is a product state~\cite{Lindblad1973}.
The function $\rho \mapsto I(A_1\!:\!...\!:\!A_m)_{\rho}$ is lower semicontinuous on $\D(\H_{A_1..A_m})$ and takes values in $[0,+\infty]$~\cite{HOLEVO-CHANNELS-2,Lindblad1973}.
\medskip

A \emph{quantum operation} $\Phi$ from a system $A$ to a system
$B$ is a completely positive trace-non-increasing linear map from
$\T(\HH_A)$ into $\T(\HH_B)$. A trace preserving quantum operation is called  \emph{quantum channel}~\cite{HOLEVO-CHANNELS-2,MARK}.   \smallskip

In analysis of infinite-dimensional quantum systems the notion of strong convergence of quantum operations is widely used~\cite{HOLEVO-CHANNELS-2}.
A sequence $\{\Phi_n\}$ of quantum operations  from $A$ to $B$ \emph{strongly converges} to a quantum operation $\Phi_0$ if
\begin{equation*}
\lim_{n\rightarrow+\infty}\Phi_n(\rho)=\Phi_0(\rho)\quad \forall \rho\in\D(\HH_A).
\end{equation*}

We will use the following simple observation. \smallskip

\begin{lemma}\label{s-c-r} \emph{If $\{\Phi_n\}$ is a sequence of quantum operations  from $A$ to $B$ strongly converging to a quantum operation $\Phi_0$ then
\begin{equation*}
\lim_{n\rightarrow+\infty}\Phi_n(\rho_n)=\Phi_0(\rho_0)
\end{equation*}
for any sequence $\{\rho_n\}\subset\T(\HH_A)$ converging to an operator $\rho_0$.}
\end{lemma}

\section{Convergence criterion for the relative entropy distance to a convex set of states}

Let $\FF$ be a convex subset of $\D(\HH)$ and
\begin{equation*}
D_{\FF}(\rho)\doteq\inf_{\sigma\in\FF}D(\rho\shs\|\shs\sigma).
\end{equation*}

Parts~A and~B of the following theorem give, respectively, a simple sufficient condition and a criterion  of local continuity of the function $D_{\FF}$. Denote by $\widetilde{\FF}$ the cone in $\T(\H)$ generated by the set $\FF$.\smallskip

\medskip
\begin{theorem}\label{conv-th} \emph{Let $\FF$ be a convex subset of $\D(\HH)$ such that the cone $\widetilde{\FF}$ is weak* closed.}\footnote{The verifiable conditions implying the validity of this requirement are presented in~\cite[Th.7]{achievability}.}
\it Let $(\rho_n)_{n\in \N}$ be a sequence of states
in $\D(\HH)$ converging to a state $\rho_0$.
\begin{enumerate}[A)]
\item [$\rm A)$] If there exists a sequence $(\omega_n)_{n\in \N}\subset\FF$ converging to a state $\omega_0\in\FF$  such that
\begin{equation}\label{re-c}
\lim_{n\to+\infty}D(\rho_n\|\shs\omega_n)=D(\rho_0\|\shs\omega_0)<+\infty
\end{equation}
then
\begin{equation}\label{rd-conv}
\lim_{n\to+\infty}D_{\FF}(\rho_n)= D_{\FF}(\rho_0)<+\infty.
\end{equation}
If the state $\rho_0$ is faithful then limit relation~\eqref{rd-conv} implies the existence of a converging sequence $(\omega_n)_{n\in \N}\subset\FF$ such that
\eqref{re-c} holds.\footnote{Necessary and sufficient conditions for the validity of \eqref{re-c} for given converging sequences $(\rho_n)_{n\in \N}$ and $(\omega_n)_{n\in \N}$ can be found in \cite{REC-cor} and in \cite{DTL} (Propositions 2 and 3, Corollary 3).}

\item [$\rm B)$] If  the state $\rho_0$ is arbitrary then relation~\eqref{rd-conv} holds  if and only if the sequence $(\rho_n)_{n\in \N}$ has the following property: for any subsequence $(\rho_{n_k})_{k\in \N}$ there exist a sub-subsequence $(\rho_{n_{k_t}})_{t\in \N}$ and a sequence $(\omega_t)_{t\in \N}\subset\FF$ converging to a state $\omega_0\in\FF$  such that
\begin{equation}\label{re-c+}
\lim_{t\to+\infty}D(\rho_{n_{k_t}}\|\shs\omega_t)=D(\rho_0\|\shs\omega_0)<+\infty.
\end{equation}
\end{enumerate}
\end{theorem}

\begin{remark}\label{new-r-1}
Theorem~\ref{conv-th}A and claim~a) of Theorem~5 in~\cite{achievability}
show that the validity of~\eqref{re-c} for a converging sequence $(\omega_n)_{n\geq0}$ of \emph{arbitrary} free states implies the validity of~\eqref{re-c} for any sequence $(\omega_n)_{n\geq0}$ consisting of \emph{optimal} free states, i.e. such that $D(\rho_n\|\shs\omega_n)=D_{\FF}(\rho_n)$, $n\geq 0$.  Any sequence $(\omega_n)_{n>0}$ of optimal free states converges to a unique optimal free state $\omega_0$ provided that the state $\rho_0$ is faithful (see the proof of Theorem~\ref{conv-th}A below).
\end{remark}\medskip

\begin{remark}\label{new-r-2} Without  the condition of weak* closedness of  the cone $\widetilde{\FF}$ the proof below allows us to show
that~\eqref{re-c} implies that
\begin{equation*}
\limsup_{n\to+\infty}D_{\FF}(\rho_n)\leq D_{\FF}(\rho_0).
\end{equation*}
\end{remark}\medskip

\begin{proof}
We prove the claims one by one.
\begin{enumerate}[A)]
\item Assume that  relation~\eqref{re-c} holds for a sequence $(\omega_n)_n\subset\FF$ converging to a state $\omega_0\in\FF$. Since $D(\rho_0\|\shs\omega_0)<+\infty$
implies $D_{\FF}(\rho_0)<+\infty$, to prove~\eqref{rd-conv} it suffices,  by claim b) of Theorem 5 in~\cite{achievability}, to show that
\begin{equation}\label{D-lr}
\limsup_{n\to+\infty}D_{\FF}(\rho_n)\leq D_{\FF}(\rho_0).
\end{equation}

For given $\lambda\in(0,1)$ and $n\geq0$ consider the quantity
$$
Y_{\lambda,n}=\inf_{\sigma\in\FF}D(\rho_n\shs\|\shs (1-\lambda)\sigma+\lambda\omega_n).
$$
The convexity of $\FF$ implies that $Y_{\lambda,n}\geq D_{\FF}(\rho_n)$, because $(1-\lambda)\sigma+\lambda\omega_n\in \FF$ for all $n$. On the other hand, by the joint convexity of the relative entropy we see that 
$$
Y_{\lambda,n}\leq(1-\lambda)\inf_{\sigma\in\FF}D(\rho_n\shs\|\shs\sigma)+\lambda D(\rho_n\shs\|\,\omega_n)=(1-\lambda)D_{\FF}(\rho_n)+\lambda D(\rho_n\shs\|\,\omega_n).
$$
Since up to ignoring a finite number of indices we may assume that $D(\rho_n\|\shs\omega_n)$, which converges to $D(\rho_0\|\shs\omega_0)<+\infty$ as $n\to+\infty$, is finite for all $n$, we have that
\begin{equation} \label{inf-lambda}
D_{\FF}(\rho_n)=\inf_{\lambda\in (0,1)}Y_{\lambda,n}\quad \forall n\geq0.
\end{equation}
Thus, by swapping the infimum and the $\limsup$ using the elementary relation
\begin{equation} \label{elementary-swap}
\begin{aligned}
\limsup_{n\to\infty} \inf_\lambda a_{\lambda,n} &= \inf_{N} \sup_{n\geq N} \inf_\lambda a_{\lambda,n} \\
&\leq \inf_{N} \inf_\lambda \sup_{n\geq N} a_{\lambda,n} \\
&= \inf_\lambda \inf_{N} \sup_{n\geq N} a_{\lambda,n} \\
&= \inf_{\lambda} \limsup_{n\to+\infty} a_{\lambda,n} ,
\end{aligned}
\end{equation}
valid for any sequence $(a_{\lambda,n})_n$ of functions of $\lambda$, we obtain that
\begin{equation} \label{D-opt-swap}
\limsup_{n\to+\infty}D_{\FF}(\rho_n) = \limsup_{n\to+\infty} \inf_{\lambda\in (0,1)} Y_{\lambda,n} \leq \inf_{\lambda\in (0,1)} \limsup_{n\to+\infty} Y_{\lambda, n} .
\end{equation}
Now, if we could prove that
\begin{equation} \label{D-lr+}
\limsup_{n\to+\infty}Y_{\lambda,n}\leq Y_{\lambda,0}\quad \forall\lambda\in(0,1).
\end{equation}
then~\eqref{D-lr} would follow immediately from~\eqref{D-opt-swap} and~\eqref{inf-lambda}. Eq.~\eqref{D-lr+} can be proved by showing that
\begin{equation} \label{D-final}
\lim_{n\to+\infty}D(\rho_n\shs\|\shs(1-\lambda)\sigma+\lambda\omega_n)=D(\rho_0\shs\|\shs(1-\lambda)\sigma+\lambda\omega_0)<+\infty
\end{equation}
for any state $\sigma$ in $\FF$. Indeed, once~\eqref{D-final} has been established we can write
\begin{equation*} \begin{aligned}
\limsup_{n\to+\infty} Y_{\lambda,n} &= \limsup_{n\to+\infty} \inf_{\sigma\in\FF} D(\rho_n\shs\|\shs (1-\lambda)\sigma+\lambda\omega_n) \\
&\leq \inf_{\sigma\in \FF} \limsup_{n\to+\infty} D(\rho_n\shs\|\shs (1-\lambda)\sigma+\lambda\omega_n) \\
&= \inf_{\sigma\in \FF} D(\rho_0\shs\|\shs (1-\lambda)\sigma+\lambda\omega_0) \\
&= Y_{\lambda, 0} ,
\end{aligned} \end{equation*}
where the inequality follows once again from~\eqref{elementary-swap}. Since $\lambda\omega_n\leq(1-\lambda)\sigma+\lambda\omega_n$, 
Eq.~\eqref{D-final} follows from~\eqref{re-c} by~\cite[Proposition~2]{DTL}.

Assume that $\rho_0$ is a faithful state. By~\cite[Theorem~5, claim~a)]{achievability} for each $n$ there is a state $\omega_n$ in $\FF$ such that $D_{\FF}(\rho_{n})=D(\rho_n\shs\|\shs\omega_n)$. By~\cite[Proposition~18]{achievability} relation~\eqref{rd-conv} implies that the sequence $(\omega_n)_{n\in \N}$ converges to the unique state $\omega_0$ in $\FF$ such that $D_{\FF}(\rho_{0})=D(\rho_0\shs\|\shs\omega_0)$. This shows the validity of~\eqref{re-c} for the sequence $(\omega_n)_{n\in \N}$ and the state $\omega_0$.\smallskip

\item To prove the ``only if'' part of claim~B  assume that  relation~\eqref{rd-conv} holds and $(\rho_{n_k})_{k\in \N}$ is an arbitrary subsequence of the sequence $(\rho_n)_{n\in \N}$. By~\cite[Theorem~5, claim~a)]{achievability} for each $k$ there is a state $\omega_k$ in $\FF$ such that $D_{\FF}(\rho_{n_k})=D(\rho_{n_k}\|\shs\omega_k)$. By~\cite[Remark~19]{achievability} there is a subsequence $(\omega_{k_l})_{l\in \N}$ converging to a state $\omega_0\in\FF$ such that
$$
\lim_{l\to+\infty}D(\rho_{n_{k_l}}\|\shs\omega_{k_l})=D(\rho_0\|\shs\omega_0)=D_{\FF}(\rho_0)<+\infty.
$$
The ``if'' part of claim~B can be derived from claim~A by using Lemma~\ref{t-lemma} below. 
\end{enumerate}
\end{proof}

\medskip
\begin{lemma}\label{t-lemma} \emph{A sequence $(x_n)_{n\in \N}\subset \mathbb{R}$ converges to $x_0\in\mathbb{R}$ if and only if for any subsequence $(x_{n_k})_{k\in \N}$ there exist a sub-subsequence $(x_{n_{k_l}})_{l\in \N}$ converging to $x_0$.}
\end{lemma}

\medskip
Theorem~\ref{conv-th}A implies the following.

\medskip
\begin{corollary}\label{conv-th-c-1} \emph{Let $\FF$ be a convex subset of $\D(\HH)$ such that the cone $\widetilde{\FF}$ is weak* closed.
Let $(\rho_n)_{n\in \N}$ be a sequence of states
in $\D(\HH)$ converging to a state $\rho_0$. If there exists a sequence $(\omega_n)_{n\in \N}\subset\FF$ converging to a state $\omega_0\in\FF$  such that
\begin{equation}\label{re-c++}
\lim_{n\to+\infty}\Tr \rho_n(-\ln\omega_n)=\Tr \rho_0(-\ln\omega_0)<+\infty
\end{equation}
then}
\begin{equation*}
\lim_{n\to+\infty}D_{\FF}(\rho_n)= D_{\FF}(\rho_0)<+\infty.
\end{equation*}
\end{corollary} 

\begin{proof}
It suffices to note that condition~\eqref{re-c++} implies relation
\eqref{re-c}. This follows from the lower semicontinuity of the entropy and the relative entropy, since representation~\eqref{re-exp} shows that
$$
\Tr \rho_n(-\ln\omega_n)=D(\rho_n\|\shs\omega_n)+S(\rho_n)\quad \forall n\geq0.\;\Box
$$

To formulate basic corollaries of Theorem~\ref{conv-th}B introduce the following notation. For a given subset $\FF$ of $\D(\HH)$
denote by $\Upsilon_{\FF}$ the set of all converging sequences $(\rho_n)_{n\in\N}$ such that
$$
\lim_{n\to+\infty}D_{\FF}(\rho_n)=D_{\FF}(\rho_0)<+\infty,
$$
where $\rho_0$ is the limit of $(\rho_n)_{n\in\N}$.\smallskip

In the following corollary we assume that $\FF$ (correspondingly, $\FF_X$, $X=1,2,A,B,AB$) are convex subsets of $\D(\HH)$ (correspondingly, of $\D(\HH_X)$) such that the cones $\widetilde{\FF}$ (correspondingly, $\widetilde{\FF}_X$, $X=1,2,A,B,AB$) generated by these subsets are weak* closed. Speaking about converging
sequences $(\rho_n)_{n\geq0}$ (correspondingly, $(\sigma_n)_{n\geq0}$, $(\omega_n)_{n\geq0}$) we always assume that they converge to the states $\rho_0$
(correspondingly, $\sigma_0$, $\omega_0$).
\end{proof}

\medskip
\begin{corollary}\label{conv-th-c-2} 
\emph{The following holds:}
\begin{enumerate}[A)]
\item \emph{If $(\rho_n)_{n\geq0}\subset\D(\H)$ is a converging sequence such that $\,c\rho_n\leq \sigma$ holds for all $n\geq 0$ and some $c>0$, where $\sigma$ is a state such that $D_{\FF}(\sigma)<+\infty$, then $(\rho_n)_{n\geq0}\in \Upsilon_{\FF}$.} 

\item \emph{If $(\rho_n)_{n\geq0}\subset\D(\H)$ is a converging sequence such that $c\rho_n\leq \sigma_n$ for all $n\geq 0$ and some $c>0$, where $(\sigma_n)_{n\geq0}$ is a sequence in $\Upsilon_{\FF}$, then $(\rho_n)_{n\geq0}\in \Upsilon_{\FF}$.} 

\item \emph{If $(\rho_n)_{n\geq0}\subset\D(\H)$ and $(\sigma_n)_{n\geq0}\subset\D(\H)$ are converging sequences of states and $(p_n)_{n\geq0}$ is a sequence of numbers in $[0,1]$ converging to $p_0\in[0,1]$, then
\begin{equation}\label{t-eqv}
\{(\rho_n)_{n\geq0}\in \Upsilon_{\FF}\}\wedge\{(\sigma_n)_{n\geq0}\in \Upsilon_{\FF}\}\quad\Longrightarrow\quad \left\{(p_n\rho_n+(1-p_n)\sigma_n)_{n\geq0}\in \Upsilon_{\FF}\right\}.
\end{equation}
If $p_0\in(0,1)$ then ``$\Longleftrightarrow$'' holds in~\eqref{t-eqv}.} 

\item \emph{If  $\,\FF_A$ and $\FF_B$ are subsets of $\,\D(\H_A)$ and $\D(\H_B)$ correspondingly then
$$
\{(\rho_n)_{n\geq0}\in \Upsilon_{\FF_A}\}\quad\Longrightarrow\quad \left\{(\Phi(\rho_n))_{n\geq0}\in \Upsilon_{\FF_B}\right\}
$$
for any quantum channel $\,\Phi:A\mapsto B$ such that $\,\Phi(\FF_A)\subseteq\FF_B$.} 

\item \emph{If  $\,\FF_A$ and $\FF_{B}$ are subsets of $\D(\H_A)$ and $\D(\H_B)$, 
respectively, then
$$
\{(\rho_n)_{n\geq0}\in \Upsilon_{\FF_A}\}\quad\Longrightarrow\quad \left\{([\Phi_n(\rho_n)])_{n\geq0}\in \Upsilon_{\FF_B}\right\}
$$
for any sequence of quantum operations $\,\Phi_n:A\mapsto B$ strongly converging to a quantum operation $\,\Phi_0$ such that $\,\Phi_n(\FF_A)\subseteq\widetilde{\FF}_B$
and $\,c_n\doteq\Phi_n(\rho_n)\neq0$ for all $n\geq0$, where $[\Phi_n(\rho_n)]$ denotes the state $c_n^{-1}\Phi_n(\rho_n)$.} 

\item \emph{If  $\,\FF_1\subseteq\FF_2$ then $\{(\rho_n)_{n\geq0}\in \Upsilon_{\FF_1}\}\;\Longrightarrow\;\{(\rho_n)_{n\geq0}\in \Upsilon_{\FF_2}\}$.}\smallskip

\item \emph{Let  $\,\FF_A$, $\FF_B$ and $\FF_{AB}$ be  subsets of $\,\D(\H_A)$, $\D(\H_B)$ and $\D(\H_{AB})$ correspondingly.} 

\emph{If $\,\omega_A\otimes\omega_B\in\FF_{AB}$ for any $\,\omega_A\in\FF_{A}$ and $\,\omega_B\in\FF_{B}$ then the implication
$$
\{(\rho^A_n)_{n\geq0}\in \Upsilon_{\FF_A}\}\wedge\{(\rho^B_n)_{n\geq0}\in \Upsilon_{\FF_B}\}\quad\Longrightarrow\quad \left\{(\rho^{AB}_n)_{n\geq0}\in \Upsilon_{\FF_{AB}}\right\}
$$
holds for any converging sequence $(\rho^{AB}_n)_{n\geq0}\subset\D(\H_{AB})$ such that\footnote{$I(A\!:\!B)$ denotes the quantum mutual information defined in~\eqref{MI-m-def}.}
\begin{equation}\label{QMI-cont}
\lim_{n\to+\infty}I(A\!:\!B)_{\rho_n}=I(A\!:\!B)_{\rho_0}<+\infty.
\end{equation}}

\emph{If $\,\omega_X\in\FF_{X}$ for any $\,\omega_{AB}\in\FF_{AB}$, where $X$ is either $A$ or $B$, then the implication
$$
\left\{(\rho^{AB}_n)_{n\geq0}\in \Upsilon_{\FF_{AB}}\right\}\quad\Longrightarrow\quad\{(\rho^X_n)_{n\geq0}\in \Upsilon_{\FF_X}\}
$$
holds for arbitrary converging sequence $(\rho^{AB}_n)_{n\geq0}\subset\D(\H_{AB})$.}
\end{enumerate}
\end{corollary}

\begin{proof}
We prove the above claims one at a time.
\begin{enumerate}[A)]
\item Follows immediately from 
claim B proved below. 
\item Let $(\rho_{n_k})_{k\geq0}$ be an arbitrary subsequence of  $(\rho_n)_{n\geq0}$.  Since $(\sigma_n)_{n\geq0}\in\Upsilon_{\FF}$, Theorem~\ref{conv-th}B guarantees the existence of a subsequence $(\sigma_{n_{k_t}})_{t\geq0}$ and a sequence $(\omega_t)_{t>0}\subset\FF$ converging to a state $\omega_0\in\FF$  such that
$$
\lim_{t\to+\infty}D(\sigma_{n_{k_t}}\|\shs\omega_t)=D(\sigma_0\|\shs\omega_0)<+\infty.
$$
The condition $c\rho_n\leq \sigma_n$ implies, by Proposition~2 and Remark~4 in~\cite{DTL}, the validity of relation~\eqref{re-c+}. So, $(\rho_n)_{n\geq0}\in\Upsilon_{\FF}$
by Theorem~\ref{conv-th}B. 

\item Assume that $(\rho_n)_{n\geq0},(\sigma_n)_{n\geq0}\in \Upsilon_{\FF}$. Let $\tau_n=p_n\rho_n+(1-p_n)\sigma_n$ for all $n\geq 0$ and
$(\tau_{n_k})_{k\geq0}$ be an arbitrary subsequence of  $(\tau_n)_{n\geq0}$. Since $(\rho_n)_{n\geq0}\in\Upsilon_{\FF}$, Theorem~\ref{conv-th}B guarantees the existence of a subsequence $(\rho_{n_{k_t}})_{t\geq0}$ and a sequence $(\omega^{\rho}_t)_{t>0}\subset\FF$ converging to a state $\omega^{\rho}_0\in\FF$  such that
\begin{equation}\label{rho-lr}
\lim_{t\to+\infty}D(\rho_{n_{k_t}}\|\shs\omega^{\rho}_t)=D(\rho_0\|\shs\omega^{\rho}_0)<+\infty.
\end{equation}
Since $(\sigma_n)_{n\geq0}\in\Upsilon_{\FF}$, Theorem~\ref{conv-th}B guarantees the existence of a subsequence $(\sigma_{n_{k_{t_s}}})_{s\geq0}$ and a sequence $(\omega^{\sigma}_s)_{s>0}\subset\FF$ converging to a state $\omega^{\sigma}_0\in\FF$  such that
\begin{equation}\label{sigma-lr}
\lim_{s\to+\infty}D(\sigma_{n_{k_{t_s}}}\|\shs\omega^{\sigma}_s)=D(\sigma_0\|\shs\omega^{\sigma}_0)<+\infty.
\end{equation}
Limit relations~\eqref{rho-lr} and~\eqref{sigma-lr} imply, by Propositions~2 and~3 in~\cite{DTL}, that
$$
\lim_{s\to+\infty}D(\tau_{n_{k_{t_s}}}\|\shs\textstyle\frac{1}{2}(\omega^{\rho}_{t_s}+\omega^{\sigma}_s))=D(\tau_0\|\shs\frac{1}{2}(\omega^{\rho}_0+\omega^{\sigma}_0))<+\infty.
$$
Since  $(\frac{1}{2}(\omega^{\rho}_{t_s}+\omega^{\sigma}_s))_{s\geq0}\subset\FF$ by the convexity of $\FF$, $(\tau_n)_{n\geq0}\in\Upsilon_{\FF}$ by Theorem~\ref{conv-th}B. 

The last claim of C follows from part~B of the corollary. 

\item 
Follows immediately from claim~E proved below. 

\item Denote the state $[\Phi_n(\rho_n)]$ by $\sigma_n$ for all $n\geq0$. By Lemma~\ref{s-c-r} in Section 1 the sequence $(\sigma_{n})_{n>0}$
converges to the state $\sigma_0$. Let $(\sigma_{n_k})_{k\geq0}$ be an arbitrary subsequence of  $(\sigma_n)_{n\geq0}$. Since $(\rho_n)_{n\geq0}\in\Upsilon_{\FF_A}$, Theorem~\ref{conv-th}B guarantees the existence of a subsequence $(\rho_{n_{k_t}})_{t\geq0}$ and a sequence $(\omega_t)_{t>0}\subset\FF_A$ converging to a state $\omega_0\in\FF_A$  such that
$$
\lim_{t\to+\infty}D(\rho_{n_{k_t}}\|\shs\omega_t)=D(\rho_0\|\shs\omega_0)<+\infty.
$$
By Corollary 2 in~\cite{REC-cor} we have
\begin{equation}\label{a-l-r}
\lim_{t\to+\infty}D(\Phi_{n_{k_t}}(\rho_{n_{k_t}})\|\shs\Phi_{n_{k_t}}(\omega_t))=D(\Phi_{0}(\rho_0)\|\shs\Phi_{0}(\omega_0))<+\infty.
\end{equation}
Since $\Phi_n(\rho_n)\neq0$ for all $n\geq0$, we may assume that $\,\Phi_{n_{k_t}}(\omega_t)\neq0\,$ for all $t$ and that $\Phi_{0}(\omega_0)\neq0$. Let $\tilde{\omega}_t$  and $\tilde{\omega}_0$ be the states proportional to the operators $\Phi_{n_{k_t}}(\omega_t)$ and
$\Phi_{0}(\omega_0)$. By Lemma~\ref{s-c-r} in Section 1 the sequence $(\tilde{\omega}_t)_{t>0}\subset\FF_B$ converges to the state $\tilde{\omega}_0\in\FF_B$. By using identities~\eqref{D-mul} and~\eqref{D-c-id} it is easy to show that limit relation~\eqref{a-l-r} implies that
$$
\lim_{t\to+\infty}D(\sigma_{n_{k_t}}\|\shs\tilde{\omega}_t)=D(\sigma_{0}\|\shs\tilde{\omega}_0)<+\infty.
$$
Thus, $(\sigma_n)_{n\geq0}\in\Upsilon_{\FF_A}$  by  Theorem~\ref{conv-th}B.

\item 
Follows straightforwardly from Theorem~\ref{conv-th}B. 

\item To prove the first claim assume that $(\rho_n)_{n\geq0}$ is a converging sequence satisfying~\eqref{QMI-cont} such that $(\rho^A_n)_{n\geq0}\in \Upsilon_{\FF_A}$ and $(\rho^B_n)_{n\geq0}\in \Upsilon_{\FF_B}$. Let  $(\rho_{n_k})_{k\geq0}$ be an arbitrary subsequence of this sequence.  Since $(\rho^A_n)_{n\geq0}\in \Upsilon_{\FF_A}$, Theorem~\ref{conv-th}B guarantees the existence of a subsequence $(\rho^A_{n_{k_t}})_{t\geq0}$ and a sequence $(\omega^A_t)_{t>0}\subset\FF_A$ converging to a state $\omega^A_0\in\FF_A$  such that
\begin{equation}\label{rho-lr+}
\lim_{t\to+\infty}D(\rho^A_{n_{k_t}}\|\shs\omega^A_t)=D(\rho^A_0\|\shs\omega^A_0)<+\infty.
\end{equation}
Since $(\rho^B_n)_{n\geq0}\in\Upsilon_{\FF_B}$, Theorem~\ref{conv-th}B guarantees the existence of a subsequence $(\rho^B_{n_{k_{t_s}}})_{s\geq0}$ and a sequence $(\omega^B_s)_{s>0}\subset\FF_B$ converging to a state $\omega^B_0\in\FF_B$  such that
\begin{equation}\label{sigma-lr+}
\lim_{s\to+\infty}D(\rho^B_{n_{k_{t_s}}}\|\shs\omega^B_s)=D(\rho^B_0\|\shs\omega^B_0)<+\infty.
\end{equation}

By using Lemma~\ref{D-AB} below we obtain  from~\eqref{QMI-cont},~\eqref{rho-lr+} and~\eqref{sigma-lr+} that
$$
\lim_{s\to+\infty}D(\rho_{n_{k_{t_s}}}\|\shs\omega^A_{t_s}\otimes\omega^B_s)=D(\rho_{0}\|\shs\omega^A_0\otimes\omega^B_0)<+\infty.
$$
Since $(\omega^A_{t_s}\otimes\omega^B_s)_{s\geq0}\subset\FF_{AB}$ by the condition, $(\rho_n)_{n\geq0}\in\Upsilon_{\FF_{AB}}$ by Theorem~\ref{conv-th}B. 

The second claim of G directly follows from part D of the corollary proved before.
\end{enumerate}
\end{proof}

\begin{lemma}\label{D-AB} \emph{For arbitrary states $\rho\in\D(\H_{AB})$, $\omega_A\in\D(\H_A)$  and  $\omega_B\in\D(\H_B)$ the
following identity holds
\begin{equation}\label{f-iden}
D(\rho\shs\|\shs\omega_A\otimes\omega_B)=D(\rho_A\|\shs\omega_A)+D(\rho_B\|\shs\omega_B)+I(A\!:\!B)_{\rho}
\end{equation}
in which both side may take the value $+\infty$.}\smallskip
\end{lemma}

\begin{proof}
If $\rho$ is a state such that $S(\rho)$, $S(\rho_A)$ and $S(\rho_B)$ are finite then identity~\eqref{f-iden} is proved easily by using representation~\eqref{re-exp}.

If $\rho$ is an arbitrary state then it can be approximated by the sequence of states
$$
\rho_n=\Phi_n\otimes\Psi_n(\rho),
$$
where $\{\Phi_n\}$ and $\{\Psi_n\}$ are sequences of channels from $A$ to $A$ and from $B$ to $B$ with finite-dimensional
output strongly converging to the identity channels $\id_A$ and $\id_B$ correspondingly. Since
$S(\rho_n)$, $S(\rho_n^A)$ and $S(\rho_n^B)$ are finite for each $n$, identity~\eqref{f-iden}
holds with $\rho$, $\omega_A$ and $\omega_B$ replaced, respectively, by $\rho_n$, $\Phi_n(\omega_A)$ and $\Psi_n(\omega_B)$. Thus, to complete the proof it suffices to note that
$$
\lim_{n\to+\infty}D(\rho_n\|\shs\Phi_n(\omega_A)\otimes\Psi_n(\omega_B))=
D(\rho\shs\|\shs\omega_A\otimes\omega_B)\leq+\infty,
$$
$$
\lim_{n\to+\infty}D(\rho_n^A\|\shs\Phi_n(\omega_A))=
D(\rho_A\|\shs\omega_A)\leq+\infty,\quad \lim_{n\to+\infty}D(\rho_n^B\|\shs\Psi_n(\omega_B))=
D(\rho_B\|\shs\omega_B)\leq+\infty
$$
and
$$
\lim_{n\to+\infty}I(A\!:\!B)_{\rho_n}=I(A\!:\!B)_{\rho}\leq+\infty.
$$
These limit relations follow from the lower semicontinuity of the quantum relative entropy and the fact that it is 
non-increasing under quantum channels.
\end{proof}

\section{Application: the relative entropy of multipartite 
entanglement}

In this section we apply Theorem~\ref{conv-th} and its corollaries to obtain continuity conditions for  the relative entropy of $\pi$-entanglement of a state of $m$-partite system $A_1\ldots A_m$ corresponding to any given (non-empty) set $\pi\subseteq P(m)$ of partitions of $\{1,\ldots, m\}$. For a given $\pi=\{\pi_k\}$
we define the set of $\pi$-separable states by (cf.\cite{Szalay2015})
$$
\mathcal{S}_{A_1\ldots A_m}^{\pi}\doteq \mathrm{cl}_{\mathrm{tn}} \left( \mathrm{conv}\left( \bigcup\nolimits_k \left\{ \bigotimes\nolimits_j \psi^{(j)}_{A_{\pi_k(j)}}: |\psi^{(j)}\rangle_{A_{\pi_k(j)}}\!\in \HH_{A_{\pi_k(j)}}, \langle\psi^{(j)} | \psi^{(j)}\rangle = 1 \right\} \right) \right),
$$
where $A_{\pi_k(j)}$ is the system obtained by joining those $A_i$ such that $i\in \pi_k(j)$, $\mathrm{cl}_{\mathrm{tn}}$ denotes the closure w.r.t. the trace norm and we used the shorthand notation $\psi^{(j)}\doteq |\psi^{(j)}\rangle\langle\psi^{(j)}|$ (see details in~\cite[Section V-B]{achievability}).

Let $\rho$ be a state of a $m$-partite quantum system $A_1\ldots A_m$. For a generic non-empty $\pi\subseteq P(m)$, the relative entropy of $\pi$-entanglement
of $\rho$ is defined by
$$
E_{R,\pi}(\rho) =\inf_{\sigma\in\mathcal{S}_{A_1\ldots A_m}^\pi}D(\rho\shs\|\shs\sigma).
\label{m-E-r}
$$

If the set $\pi$ contains only the finest partition $\left\{\{1\},\ldots, \{m\}\right\}$ then $\mathcal{S}_{A_1\ldots A_m}^\pi$
coincides with the convex set $\mathcal{S}_{A_1...A_m}$ of all (fully) separable states in $\D(\HH_{A_1...A_m})$ and hence  $E_{R,\pi}$ is
the relative entropy of entanglement $E_R$ in this case. It is shown in~\cite{achievability} that the convex set $\mathcal{S}_{A_1\ldots A_m}^\pi$ is weak* closed for any set of partitions $\pi$. So, since $\mathcal{S}\subseteq \mathcal{S}_{A_1\ldots A_m}^\pi$ for any $\pi$, Corollary~\ref{conv-th-c-2}F  implies the following statement. 
\medskip

\begin{property}\label{ER-conv-1} \emph{Let $(\rho_n)_{n\in\N}$ be a sequence of states in $\D(\HH_{A_1...A_m})$ converging to a state $\rho_0$ and $\pi$ any set of partitions. Then}
\begin{equation*}
\left\{\lim_{n\to+\infty} E_R(\rho_n)=E_R(\rho_0)<+\infty\right\}\quad\Longrightarrow\quad \left\{\lim_{n\to+\infty} E_{R,\pi}(\rho_n)=E_{R,\pi}(\rho_0)<+\infty\right\}.
\end{equation*}
\end{property} 
\medskip

By Proposition~\ref{ER-conv-1}, the continuity analysis of $E_{R,\pi}$ 
reduces to the continuity analysis of $E_R$. The following proposition gives several sufficient conditions for local continuity of the latter function. 

\medskip
\begin{property}\label{ER-conv-2} \emph{Let $(\rho_n)_{n\in\N}$ be a sequence of states in $\D(\HH_{A_1...A_m})$ converging to a state $\rho_0$. Then the relation
\begin{equation}\label{CS-r}
\lim_{n\to+\infty} E_R(\rho_n)=E_R(\rho_0)<+\infty
\end{equation}
holds provided that one of the following conditions is valid:}\footnote{$I(A_1\!:...:\!A_m)$ denotes the quantum mutual information defined in~\eqref{MI-m-def}.}
\begin{enumerate}[a)]
\item  \begin{equation}\label{mi-cont}
\lim_{n\to+\infty} I(A_1\!:...:\!A_m)_{\rho_n}=I(A_1\!:...:\!A_m)_{\rho_0}<+\infty;
\end{equation}
\item \emph{$c\rho_n\leq \sigma_n$ for  all $\shs n$, where $c>0$ and  $(\sigma_n)_{n\in\N}$ is a sequence of states in $\D(\HH)$ converging to a state $\sigma_0$ such that}
$$
\quad\lim_{n\to+\infty}E_R(\sigma_n)=E_R(\sigma_0)<+\infty;
$$
\item \emph{$\rho_n=p_n\rho^1_n+(1-p_n)\rho^2_n$, $n\geq0$, where $(\rho^1_n)_{n\in\N}$ and $(\rho^2_n)_{n\in\N}$
are sequences of states in $\D(\HH_{A_1...A_m})$ converging to states $\rho^1_0$ and $\rho^2_0$  such that
$$
\lim_{n\to+\infty}E_R(\rho^i_n)=E_R(\rho^i_0)<+\infty,\quad i=1,2,
$$
and $(p_n)_{n\in\N}$ is a sequence of numbers in $[0,1]$ converging to $p_0$.}

\end{enumerate}
\end{property}

\begin{proof}
If condition~a) is valid then relation~\eqref{re-c} holds with $\omega_n=\rho_n^{A_1}\otimes...\otimes\rho_n^{A_m}$, $n\geq0$.
So, this assertion follows directly from Theorem~\ref{conv-th}A. Condition~b) implies~\eqref{CS-r} by Corollary~\ref{conv-th-c-2}B. Finally, Condition~c) implies relation~\eqref{CS-r} by Corollary~\ref{conv-th-c-2}C. 
\end{proof}

\medskip
\begin{remark}\label{I-C-a} Condition a) in Proposition~\ref{ER-conv-2} has a clear physical interpretation: \emph{it states that local continuity of total correlation implies local continuity of the relative entropy of entanglement, which, in turn, implies local continuity of the relative entropy of $\pi$-entanglement for any set of partitions $\pi$ by Proposition~\ref{ER-conv-1}.} Similar condition was  obtained previously for the one-way classical correlation,
quantum discord, the $m$-partite squashed entanglement and c-squashed entanglement~\cite{Shirokov-review-AV}  (in the last two cases there is additional assumption that all the marginal states of the limit state $\rho_0$ have finite entropy).

Several sufficient conditions for the validity of~\eqref{mi-cont} are presented in~\cite[Section 5.2.3]{Shirokov-review-AV}. In particular, Proposition 21 in~\cite{Shirokov-review-AV} implies that relation~\eqref{mi-cont} holds if
\begin{equation*}
\lim_{n\to+\infty} S(\rho^{A_i}_n)=S(\rho^{A_i}_0)<+\infty
\end{equation*}
for at least $m-1$ indexes $i$, where $\rho^{A_i}_n$ is the marginal state of $\rho_n$ corresponding to the subsystem $A_i$.
\end{remark}\medskip

\begin{remark}\label{I-C-b} Condition b) in Proposition~\ref{ER-conv-2} implies, in particular, that
\emph{relation~\eqref{CS-r} holds provided that $c\rho_n\leq \sigma$ for all $\shs n$, where $c>0$ and  $\sigma$ is a state in $\D(\HH_{A_1...A_m})$ such that  $E_R(\sigma)$ is finite.} This assertion is similar to Simon's dominated convergence theorem for the von Neumann entropy~\cite[the Appendix]{lieb73b}.
\end{remark}\medskip

By applying Corollary~\ref{conv-th-c-2}E we obtain the following. 

\medskip
\begin{property}\label{ER-conv-3} \emph{Let $(\rho_n)_{n\in\N}$ be a sequence of states in $\D(\HH_{A_1...A_m})$ converging to a state $\rho_0$ such that relation
\eqref{CS-r} holds. Then
\begin{equation*}
\lim_{n\to+\infty} 
E_R\left(\frac{\Phi_n(\rho_n)}{\Tr\Phi_n(\rho_n)}\right) = E_R\left(\frac{\Phi_0(\rho_0)}{\Tr \Phi_0(\rho_0)}\right)<+\infty
\end{equation*}
for any sequence $\{\Phi_n\}$ of quantum operations from $A_1...A_m$ to $B_1...B_{l}$ strongly converging to a quantum operation $\Phi_0$ such that $\Phi_n(\rho_n)\neq0$ and  $\,\Phi_n(\mathcal{S}_{A_1...A_m})\subseteq \widetilde{\mathcal{S}}_{B_1...B_l}$ for all $n\geq0$, where
$\widetilde{\mathcal{S}}_{B_1...B_l}$ is the cone generated by the set $\mathcal{S}_{B_1...B_l}$ of all (fully) separable states in $\D(\HH_{B_1...B_l})$.}
\end{property} 

\medskip
\begin{remark}\label{ER-conv-r}
Continuity conditions given by Propositions~\ref{ER-conv-2} and~\ref{ER-conv-3} with $m=l=2$ are also valid
for the relative entropy distance to the PPT states, since that cone generated by the set of PPT states is weak* closed~\cite{achievability}.
\end{remark}
\medskip

Proposition~\ref{ER-conv-3} and Corollary~\ref{conv-th-c-2}G imply the following statement. 

\medskip
\begin{corollary}\label{ER-conv-3-c} \emph{Let $(\rho_n)_{n\in\N}$ be a sequence  in $\D(\HH_{A_1...A_m})$ converging to a state $\rho_0$.}\smallskip

\emph{If limit relation~\eqref{CS-r} holds then
\begin{equation}\label{CS-r+++}
\lim_{n\to+\infty} E_R(\rho^B_n)=E_R(\rho^B_0)<+\infty,
\end{equation}
where $B\doteq A_{i_1}...A_{i_l}$ is any $l$-partite subsystem of $A_1...A_m$ and $E_R$ denotes the relative entropy distance to the set
of all (fully) separable states in $\D(\HH_{A_{i_1}...A_{i_l}})$.}\smallskip

\emph{If limit relation~\eqref{CS-r+++} holds for some $l$-partite subsystem $B\doteq A_{i_1}...A_{i_l}$ of $A_1...A_m$
and for the complementary subsystem $\bar{B}\doteq A_1...A_m \setminus A_{i_1}...A_{i_l}$ and also
\begin{equation}\label{I-cond}
\lim_{n\to+\infty}I(B\!:\!\bar{B})_{\rho_n}=I(B\!:\!\bar{B})_{\rho_0}<+\infty
\end{equation}
then limit relation~\eqref{CS-r} holds.}
\end{corollary}

\medskip
The question of whether the condition~\eqref{I-cond} is necessary remains open.

\section*{Acknowledgement}

The authors are grateful to A.\ S.\ Holevo for useful discussions 
surrounding the topic of this work. LL thanks the Freie Universit\"at Berlin for hospitality during the completion of this paper.
\bibliographystyle{unsrt}
\bibliography{biblio}

\begin{thebibliography}{10}

\bibitem{RT-review}
E.~Chitambar and G.~Gour.
\newblock Quantum resource theories.
\newblock {\em Rev. Mod. Phys.}, 91:025001, 2019.

\bibitem{Vedral1997}
V.~Vedral, M.~B. Plenio, M.~A. Rippin, and P.~L. Knight.
\newblock Quantifying entanglement.
\newblock {\em Phys. Rev. Lett.}, 78:2275--2279, 1997.

\bibitem{tightuniform}
A.~Winter.
\newblock Tight uniform continuity bounds for quantum entropies: Conditional
  entropy, relative entropy distance and energy constraints.
\newblock {\em Commun. Math. Phys.}, 347(1):291--313, 2016.

\bibitem{achievability}
L.~Lami and M.~E. Shirokov.
\newblock Attainability and lower semi-continuity of the relative entropy of
  entanglement and variations on the theme.
\newblock {\em Ann. Henri Poincar{\'e}}, 2023.
\newblock https://doi.org/10.1007/s00023-023-01313-1


\bibitem{Umegaki1962}
H.~Umegaki.
\newblock {Conditional expectation in an operator algebra. IV. Entropy and
  information}.
\newblock {\em Kodai Math. Sem. Rep.}, 14(2):59--85, 1962.

\bibitem{Hiai1991}
F.~Hiai and D.~Petz.
\newblock The proper formula for relative entropy and its asymptotics in
  quantum probability.
\newblock {\em Comm. Math. Phys.}, 143(1):99--114, 1991.

\bibitem{Vedral1998}
V.~Vedral and M.~B. Plenio.
\newblock Entanglement measures and purification procedures.
\newblock {\em Phys. Rev. A}, 57:1619--1633, 1998.

\bibitem{Donald1999}
M.~J. Donald and M.~Horodecki.
\newblock Continuity of relative entropy of entanglement.
\newblock {\em Phys. Lett. A}, 264(4):257--260, 1999.

\bibitem{Plenio2000}
M.~B. Plenio, S.~Virmani, and P.~Papadopoulos.
\newblock Operator monotones, the reduction criterion and the relative entropy.
\newblock {\em J. Phys. A}, 33(22):L193--L197, 2000.

\bibitem{MatthiasPhD}
M.~Christandl.
\newblock {\em The Structure of Bipartite Quantum States - Insights from Group
  Theory and Cryptography}.
\newblock PhD thesis, University of Cambridge, 2006.
\newblock https://arxiv.org/abs/quant-ph/0604183

\bibitem{BrandaoPlenio2}
F.~G. S.~L. Brand{\~a}o and M.~B. Plenio.
\newblock A reversible theory of entanglement and its relation to the second
  law.
\newblock {\em Commun. Math. Phys.}, 295(3):829--851, 2010.

\bibitem{Brandao2010}
F.~G. S.~L. Brand{\~a}o and M.~B. Plenio.
\newblock A generalization of quantum {S}tein's lemma.
\newblock {\em Commun. Math. Phys.}, 295(3):791--828, 2010.

\bibitem{Piani2009}
M.~Piani.
\newblock Relative entropy of entanglement and restricted measurements.
\newblock {\em Phys. Rev. Lett.}, 103:160504, 2009.

\bibitem{Datta-alias}
N.~Datta.
\newblock Max-relative entropy of entanglement, alias log robustness.
\newblock {\em Int. J. Quantum Inf.}, 07(02):475--491, 2009.

\bibitem{Brandao-Gour}
F.~G. S.~L. Brand{\~a}o and G.~Gour.
\newblock Reversible framework for quantum resource theories.
\newblock {\em Phys. Rev. Lett.}, 115:070503, 2015.

\bibitem{gap}
M.~Berta, F.~G. S.~L. Brand\~{a}o, G.~Gour, L.~Lami, M.~B. Plenio, B.~Regula,
  and M.~Tomamichel.
\newblock On a gap in the proof of the generalised quantum {S}tein's lemma and
  its consequences for the reversibility of quantum resources.
\newblock {\em Preprint arXiv:2205.02813}, 2022.

\bibitem{Wehrl}
A.~Wehrl.
\newblock General properties of entropy.
\newblock {\em Rev. Mod. Phys.}, 50:221--260, 1978.

\bibitem{HOLEVO-CHANNELS-2}
A.~S. Holevo.
\newblock {\em Quantum Systems, Channels, Information: A Mathematical
  Introduction}.
\newblock Texts and Monographs in Theoretical Physics. De Gruyter, Berlin,
  Germany, 2nd edition, 2019.

\bibitem{MARK}
M.~M. Wilde.
\newblock {\em Quantum Information Theory}.
\newblock Cambridge University Press, 2nd edition, 2017.

\bibitem{PETZ-ENTROPY}
M.~Ohya and D.~Petz.
\newblock {\em Quantum Entropy and Its Use}.
\newblock Theoretical and Mathematical Physics. Springer Berlin Heidelberg,
  2004.

\bibitem{Lindblad1974}
G.~Lindblad.
\newblock Expectations and entropy inequalities for finite quantum systems.
\newblock {\em Commun. Math. Phys.}, 39(2):111--119, 1974.

\bibitem{Lindblad1973}
G.~Lindblad.
\newblock Entropy, information and quantum measurements.
\newblock {\em Commun. Math. Phys.}, 33(4):305--322, 1973.

\bibitem{REC-cor}
M.~E. Shirokov.
\newblock Convergence criterion for quantum relative entropy and its use.
\newblock {\em Sb. Math.}, 213(12):1740–1772, 2022.
\newblock {\em The PDF file is available at}  https://doi.org/10.4213/sm9794e

\bibitem{DTL}
M.~E. Shirokov.
\newblock Convergence conditions for the quantum relative entropy and other
  applications of the generalized quantum {D}ini lemma.
\newblock {\em Lobachevskii J. Math.}, 43(7):1755--1777, 2022.
\newblock {\em Preprint arXiv:2205.09108}

\bibitem{Szalay2015}
S.~Szalay.
\newblock Multipartite entanglement measures.
\newblock {\em Phys. Rev. A}, 92:042329, 2015.

\bibitem{Shirokov-review-AV}
M.~E. Shirokov.
\newblock Continuity of characteristics of composite quantum systems.
\newblock {\em Preprint arXiv:2201.11477}, 2022.

\bibitem{lieb73b}
E.~H. Lieb and M.~B. Ruskai.
\newblock Proof of the strong subadditivity of quantum mechanical entropy.
\newblock {\em J. Math. Phys.}, 14(12):1938--1941, 1973.

\end{thebibliography}

\end{document}